\documentclass{article}
\usepackage{authblk}

\usepackage[utf8]{inputenc}
\usepackage{amsmath}
\usepackage{amsfonts}
\usepackage{bbm}

\usepackage{amsthm}
\usepackage{stmaryrd}
\usepackage{tikz}
\usetikzlibrary{calc}
\usetikzlibrary{shapes}
\usetikzlibrary{arrows}
\usepackage[noend]{algpseudocode}
\usepackage{algorithm}
\usepackage{xspace}
\usepackage{caption}
\usepackage{subcaption}
\usepackage{todonotes}

\usepackage[usestackEOL]{stackengine}
\stackMath
\usepackage{graphicx}
\usepackage{placeins}

\usepackage{enumitem}

\usepackage{etoolbox}
\usepackage{mathtools}

\usepackage{soul}

\newtheorem{theorem}{Theorem}
\newtheorem{lemma}{Lemma}

\newtheorem{corollary}{Corollary}

\DeclareEmphSequence{\bfseries\itshape}

\begin{document}

\title{Maximal Line Digraphs}
\author[1]{Quentin Japhet}
\author[2,3]{Dimitri Watel}
\author[1]{Dominique Barth}
\author[4]{Marc-Antoine Weisser}

\date{quentin.japhet@uvsq.fr (Corresponding author), dimitri.watel@ensiie.fr, dominique.barth@uvsq.fr, marc-antoine.weisser@centralesupelec.fr}

\affil[1]{DAVID, University of Versailles-St Quentin, Paris-Saclay University, 45 avenue des États-Unis, 78035, Versailles, France}
\affil[2]{ENSIIE, 1 square de la Résistance, 91000, Evry, France}
\affil[3]{SAMOVAR, 9 Rue Charles Fourier, 91000, Evry, France}
\affil[4]{LISN, CentraleSupelec, Paris-Saclay University, 1 Rue Raimond Castaing, 91190, Gif-sur-Yvette, France}

\newcommand{\degsuccs}[2][]{d_{\ifthenelse{\equal{#1}{}}{}{#1,}#2}^+}
\newcommand{\degpreds}[2][]{d_{\ifthenelse{\equal{#1}{}}{}{#1,}#2}^-}
\newcommand{\degree}[2][]{d_{\ifthenelse{\equal{#1}{}}{}{#1,}#2}}
\newcommand{\preds}[2][]{\Gamma_{\ifthenelse{\equal{#1}{}}{}{#1,}#2}^-}
\newcommand{\succs}[2][]{\Gamma_{\ifthenelse{\equal{#1}{}}{}{#1,}#2}^+}
\newcommand{\neigh}[2][]{\Gamma_{\ifthenelse{\equal{#1}{}}{}{#1,}#2}}
\newcommand{\antineigh}[2][]{\Gamma_{\ifthenelse{\equal{#1}{}}{}{#1,}#2}^0}
\newcommand{\predsarcs}[2][]{\gamma_{\ifthenelse{\equal{#1}{}}{}{#1,}#2}^-}
\newcommand{\succsarcs}[2][]{\gamma_{\ifthenelse{\equal{#1}{}}{}{#1,}#2}^+}

\newcommand{\dw}[1]{\todo[color=blue!40!white]{DW : #1}}
\newcommand{\dwi}[1]{\todo[inline,color=blue!40!white]{DW : #1}}

\sethlcolor{orange}
\newcommand{\change}[1]{\begingroup \hl{#1} \endgroup}



\maketitle

\begin{abstract}
    A line digraph $L(G) = (A, E)$ is the digraph constructed from the digraph $G = (V, A)$ such that there is an arc $(a,b)$ in $L(G)$ if the terminal node of $a$ in $G$ is the initial node of $b$.
    The maximum number of arcs in a line digraph with $m$ nodes is $(m/2)^2 + (m/2)$ if $m$ is even, and $((m - 1)/2)^2 + m - 1$ otherwise.
    For $m \geq 7$, there is only one line digraph with as many arcs if $m$ is even, and if $m$ is odd, there are two line digraphs, each being the transpose of the other.
\end{abstract}


\section{Introduction}
    Introduced in \cite{beineke1968derived}, the \emph{line digraph} transformation is, given a simple digraph $G = (V, A)$ with $m$ arcs, there is an arc $(a,b)$ in the line digraph $L(G) = (A, E)$ if the terminal node of $a$ in $G$ is the initial node of $b$. We say $G$ is a \emph{root digraph} of $L(G)$. 
    The following caracterization of line digraphs is provided in \cite{beineke1968derived}.
    \begin{theorem}[\cite{beineke1968derived}, Theorem 7]
        A digraph is a line digraph if and only if none of the \emph{Shortcut}, \emph{Eight} and \emph{Deviation} digraphs shown in Figure~\ref{fig:caracterization} is a subgraph, and every $\emph{Z}$ digraph is in a $K_{2,2}$.
    \end{theorem}
    \newcommand\widthInterdit{0.24}
    \begin{figure}[ht!]
        \centering
        \hfill
        \begin{subfigure}[b]{\widthInterdit\textwidth}
            \centering
            \includegraphics[width=0.7\textwidth]{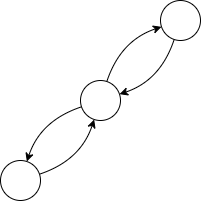}
            \caption{Eight}
            \label{fig:huit}
        \end{subfigure}
        \hfill
        \begin{subfigure}[b]{\widthInterdit\textwidth}
            \centering
            \includegraphics[width=0.7\textwidth]{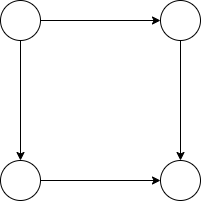}
            \caption{Deviation}
            \label{fig:deviation}
        \end{subfigure}
        \hfill
        \begin{subfigure}[b]{\widthInterdit\textwidth}
            \centering
            \includegraphics[width=0.7\textwidth]{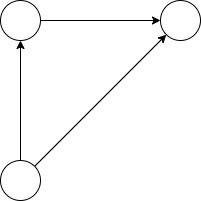}
            \caption{Shortcut}
            \label{fig:raccourci}
        \end{subfigure}
        \hfill
        \begin{subfigure}[b]{\widthInterdit\textwidth}
            \centering
            \includegraphics[width=0.7\textwidth]{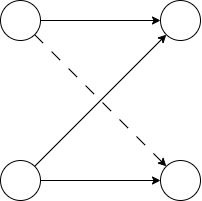}
            \caption{Z}
            \label{fig:Z}
        \end{subfigure}
        \hfill
        \captionsetup{justification=centering}
        \caption{Forbidden subdigraphs.\\The subgraph Z is allowed if the dotted arc is present, which forms $K_{2,2}$.}
        \label{fig:caracterization}
    \end{figure}
    Given a fixed value of $m$, we want to caracterize the line digraphs of $m$ nodes with maximum number of arcs. Equivalently, we want to caracterize the root digraphs with $m$ arcs maximizing $\Phi (G) = \sum_{v \in V} \degsuccs{v}\cdot \degpreds{v}$ (the outgoing and incoming degrees of $v$) as it is the number of arcs in its line digraph.
    \paragraph{Motivations}
    The undirected version of this question is trivial, the graph with the most edges, the clique, being a line graph. Previous works have been devoted to the same question for root graph with a fixed number of nodes and edges \cite{abrego2008sum, ahlswede1978graphs}. 
    The result is not obvious on the directed version of the problem, as forbidden digraphs are not induced. 
    Similar work has been done on other graph classes\cite{babinski2024maximal,hoffmann2023number,paul2009edge}.

    \paragraph{Results}In this paper, we first show in Section~\ref{sec:max_number} that the maximum number of arcs in the line digraphs with $m$ nodes is $\left(\frac{m}{2}\right)^2 + \frac{m}{2}$ arcs if $m$ is even, and $\left(\frac{m - 1}{2}\right)^2 + m - 1$ arcs otherwise. In Section~\ref{sec:unicity}, we then show that those bounds are tight by providing all the root digraphs achieving them.

    \paragraph{Notations}Like used previously, $\emph{m}$ denotes the number of arcs in the root digraphs and thus the number of node in the line digraphs while $\emph{n}$ denotes the number of nodes in the root digraphs.
    We use $\degpreds{u}$ and $\degsuccs{u}$ for the incoming and outgoing degrees of $u$. If necessary, we specify the digraph $G$ with $\degree[G]{u}$.\\
    A \emph{circuit} is a sequence of consecutive arcs whose two extremity nodes are identical.
    An arc is \emph{incident} on a node if one of its extremities is that node.\\
    Given $G = (V, A)$, $\emph{G - u}$ is the digraph $(V \setminus u, A \setminus \{(u,x), (x,u)\}_{\forall x\in V})$ and $\emph{G - (u, v)}$ is the digraph $(V, A \setminus (u, v))$.
    An \emph{optimal digraph} $G$ with $m$ arcs is such that, for all digraphs $F$ with $m$ arcs, $\Phi (F) \leq \Phi (G)$.
    The \emph{transpose digraph} $G^T = (V,A^T)$ of the digraph $G = (V,A)$ is such $(a,b) \in A \Leftrightarrow (b,a) \in A^T$.

\section{Maximum number of arcs in a line digraph}
\label{sec:max_number}

    In this section we provide a close formula for the maximum number of arcs in a line digraph of some given order. 

    \begin{lemma}
    \label{lem:maxlinedigraphedegre}
        Let $H$ be a line digraph of odd order $m \geq 7$, then $H$ contains at least one node with $\degree{v} \leq \frac{m-1}{2}$.
    \end{lemma}
    \begin{proof}
        Let $m = 2p + 1$ and $H = (A, E)$ be a line digraph with order $m$. We assume that, for every node $v_i$, $\degree{v_i} \geq p + 1$.

        As $m$ is odd, some node $v \in A$ does not belong to a 2-circuit otherwise there exists an Eight (figure~\ref{fig:caracterization}) digraph in $H$. Consequently there are at least $p + 1$ distinct neighbors to $v$. With $\preds{v}$, the set of predecessor nodes of $v$, $\succs{v}$ the set of successor nodes of $v$ and $\antineigh{v}$ the set of nodes not adjacent to $v$, we have $|\preds{v}| + |\succs{v}| \geq p+1$, $|\antineigh{v}| \leq p-1$ and $|\preds{v}| + |\succs{v}| + |\antineigh{v}| = 2p$.

        As shown in Figure~\ref{fig:Imp_2}, if two nodes of $\preds{v}$ or two nodes of $\succs{v}$ are linked, $H$ contains a Shortcut. Similarly, a node of $\preds{v}$ cannot be the predecessor of a node in $\succs{v}$. 
        Consequently, if a node of $\preds{v} \cup \succs{v}$ belongs to a 2-circuit, then the other node of the circuit is in $\antineigh{v}$. As $|\preds{v} \cup \succs{v}| - |\antineigh{v}| \geq 2$, there exist two nodes in $\preds{v} \cup \succs{v}$ that are not contained in a 2-circuit. Let $x$ be one of those nodes, there are at least $p + 1$ distinct neighbors to it.

        As shown in Figure~\ref{fig:Imp_3}, there is at most one successor in $\preds{v}$ to a node in $\succs{v}$ otherwise $H$ contains a Deviation. Similarly, there is at most one predecessor in $\succs{v}$ to a node in $\preds{v}$. Thus there is at most one neighbor $y$ in $\preds{v} \cup \succs{v}$ to $x$.

        As there are $p + 1$ distinct neighbors to $x$ and $|\antineigh{v}| \leq p - 1$, every node of $\antineigh{v}$ is neighbor of $x$. As $y$ may belong to at most one 2-circuit, then at least $p - 2$ nodes of $\antineigh{v}$ are neighbors of $y$ (and thus common neighbors of $x$ and $y$). Note that $p - 2 \geq 1$ as $m \geq 7$. Let $z$ be any of those nodes. 

        \newcommand\widthImpair{0.32}
        \begin{figure}[ht!]
            \centering
            \begin{subfigure}[b]{\widthImpair\textwidth}
                \centering
                \stackinset{c}{0\textwidth}{c}{0.15\textwidth}{v}{
                    \includegraphics[height=3.5cm]{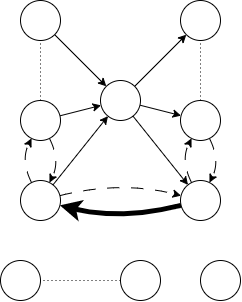}
                }
                \caption{Adding a dotted arc creates a Shortcut; adding a bold arc is possible.}
                \label{fig:Imp_2}
            \end{subfigure}
            \hfill
            \begin{subfigure}[b]{\widthImpair\textwidth}
                \centering
                \stackinset{c}{0\textwidth}{c}{0.15\textwidth}{v}{
                    \includegraphics[height=3.5cm]{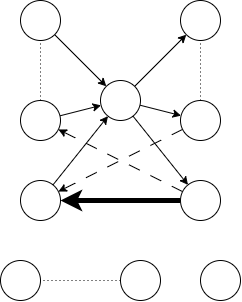}
                }
                \caption{If the bold arc is added, dotted arcs create a Deviation with the node $v$}
                \label{fig:Imp_3}
            \end{subfigure}
            \hfill
            \begin{subfigure}[b]{\widthImpair\textwidth}
                \centering
                \stackinset{c}{0\textwidth}{c}{0.15\textwidth}{v}{
                    \stackinset{c}{0.24\textwidth}{c}{-0.15\textwidth}{x}{
                        \stackinset{c}{-0.24\textwidth}{c}{-0.15\textwidth}{y}{
                            \stackinset{c}{0.055\textwidth}{b}{0.045\textwidth}{z}{
                                \includegraphics[height=3.5cm]{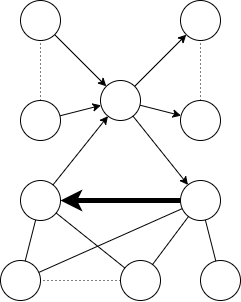}
                }}}}
                \caption{Arcs between $x$ or $y$ and $\antineigh[]{v}$ nodes are not currently oriented.}
                \label{fig:Imp_4}
            \end{subfigure}
            \captionsetup{justification=centering}
            \caption{Arcs between $\preds[]{v}$, $\succs[]{v}$ and $\antineigh[]{v}$}
            \label{fig:Imp_2_3}
        \end{figure}
        Figure~\ref{fig:Imp_4} illustrates those links when $x \in \succs{v}$. The case where $x \in \preds{v}$ is identical (and consists in swapping $x$ and $y$ on the figure). Any orientation of the links $(y, z)$ and $(x, z)$ leads either to a Shortcut or a Deviation. This last contradiction invalidates the hypothesis and prove the lemma.
    \end{proof}

    \begin{theorem}
    \label{the:maxlinedigraphe}
        In a line digraph $H$ of order $m$, there are at most $\left(\frac{m}{2}\right)^2 + \frac{m}{2}$ arcs if $m$ is even, and at most $\left(\frac{m - 1}{2}\right)^2 + m - 1$ arcs otherwise. Those bounds are tight.
    \end{theorem}

    \begin{proof}
        We prove Theorem~\ref{the:maxlinedigraphe} by induction on $m$, initialization for digraphs of order 6 or less is available 
        in the annexe.
        
        If $m$ is odd, by Lemma~\ref{lem:maxlinedigraphedegre}, there exists one node $v$ of degree at most $\frac{m - 1}{2}$. Let $H' = H - v$. By induction on the digraph order, one can show that there are at most $\left(\frac{m - 1}{2}\right)^2 + \frac{m - 1}{2}$ arcs in $H'$. As none of those arcs are incident to $v$, then there are at most $\left(\frac{m - 1}{2}\right)^2 + \frac{m - 1}{2} + \frac{m - 1}{2}$ arcs in $H$. So the theorem is true in the odd case.

        If $m$ is even. By induction on the digraph order, for every node $v$ of $H$, There are at most $B_m = \left(\frac{m - 2}{2}\right)^2 + m - 2$ arcs in the digraph $H - v$. We consider the following linear program where $x_{ij}$ is a binary variable representing the presence of the arc $(i, j)$ in the digraph. Note that no constraint is given on the fact that $H$ is a line digraph. Consequently, the program gives an upper bound on the number of arcs. We also provide the dual of the program on the right. 

        \begin{minipage}{0.49\linewidth}
            \begin{align*}
                \mathclap{\max \sum\limits_{i=1}^{m}\sum\limits_{\substack{j=1\\j \neq i}}^{m} x_{ij}^{}}\\
                \sum\limits_{\substack{j=1\\j\neq i}}^{m}\sum\limits_{\substack{k=1\\k \neq i\\k \neq j}}^{m}X_{jk}^{} &\leq B_m & \forall i \in \llbracket 1; m \rrbracket\\
                x_{ij} &\in \{0, 1\} & \forall i \neq j \in \llbracket 1; m \rrbracket
            \end{align*}
        \end{minipage}
        \begin{minipage}{0.49\linewidth}
            \begin{align*}
                \mathclap{\min \sum\limits_{\substack{i = 1\\~}}^m B_m R_i} \\
                \sum\limits_{\substack{k \neq i\\k \neq j\\~}} R_k &\leq 1 & \forall i \neq j \in \llbracket 1; m \rrbracket\\
                R_{i} &\geq 0 & \forall i \in \llbracket 1; m \rrbracket
            \end{align*}
        \end{minipage}

        A feasible primal solution for the linear relaxation consists in setting $x_{ij}$ to 
        $\frac{\left(\frac{m - 2}{2}\right)^2 + m - 2}{(m-1)(m-2)}$
        for all $i$ and $j$. A feasible dual solution consists in setting $R_i = 1/(m-2)$ for all $i$. In the two cases, we get the following objective value
        \begin{align*}
            \left( \left(\frac{m - 2}{2}\right)^2 + m - 2 \right) \cdot \frac{m}{m-2} =\left(\frac{(m - 2) \cdot m}{4}\right) + m = \left(\frac{m}{2}\right)^2 + \frac{m}{2}
        \end{align*}

        We also get the desired upper bound. The two upper bounds are tight as we can achieve them with the digraphs of Figure~\ref{fig:maxlinedigraphe}.

        \begin{figure}[h]
            \centering
            \begin{subfigure}[b]{0.3\textwidth}
                \centering
                \includegraphics[width=0.4\textwidth]{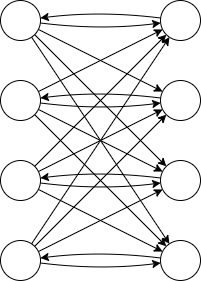}
            \end{subfigure}
            \begin{subfigure}[b]{0.3\textwidth}
                \centering
                \includegraphics[width=0.4\textwidth]{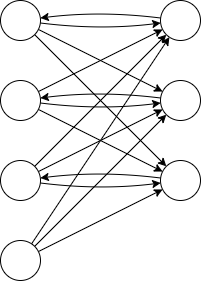}
            \end{subfigure}
            \begin{subfigure}[b]{0.3\textwidth}
                \centering
                \includegraphics[width=0.4\textwidth]{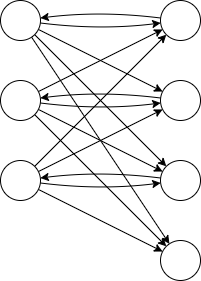}
            \end{subfigure}
            \captionsetup{justification=centering}
            \caption{Example of a line digraph with maximum number of arcs. This is a complete oriented bipartite with a maximum number of return arcs, and an additional node in the complete bipartite for the odd case.}
            \label{fig:maxlinedigraphe}
        \end{figure}

        Let see that these digraphs are line-digraphs. Since being bipartite, there is no triangle in the undirected graph underlying, and so there is no Shortcut. Each node $v$ has $\degpreds{v} = 1$ or $\degsuccs{v} = 1$, so there are no Eights either. In a Deviation, there is a node $u$ with $\degsuccs{u} > 1$, a node $v$ with $\degpreds{v} > 1$, and two paths of size two from $u$ to $v$. Since this is bipartite, the paths from a node $u$ with an $\degsuccs{u} > 1$ to $v$ with $\degpreds{v} > 1$ are all odd, so there is no Deviation. All nodes with an outgoing degree strictly greater than 1 have the same successors, and all nodes with an ingoing degree strictly greater than 1 have the same predecessors, so there is no $Z$ outside a $K_{2,2}$. This concludes the proof of Theorem~\ref{the:maxlinedigraphe}.
    \end{proof}

\section{Unicity of maximum line digraphs}
\label{sec:unicity}
    Now that we know the maximum number of arcs, let us show that the only digraphs reaching this value are those shown in Figure~\ref{fig:maxlinedigraphe}.
    Since $\Phi (G)$ is the number of arcs in $L(G)$, as a consequence of Theorem~\ref{the:maxlinedigraphe}
    \begin{corollary}
    \label{cor:maxlinedigraphe}
        Let $G$ a digraph with $m$ arcs, $\Phi(G) = \sum_{v \in V} \degsuccs{v}\cdot \degpreds{v} \leq \left(\frac{m}{2}\right)^2 + \frac{m}{2}$ arcs if $m$ is even, and $\Phi(G) \leq \left(\frac{m - 1}{2}\right)^2 + m - 1$ arcs otherwise.
    \end{corollary}
    Let $O_m$ be the digraph with a central node $w$, $\lfloor m / 2 \rfloor$ 2-length circuit from the central node and, if $m$ is odd, a final incoming arc on $w$ as shown in Figure~\ref{fig:bigroot:1}.

    \begin{figure}[h]
        \centering
        \begin{subfigure}[b]{0.3\textwidth}
            \centering
            \stackinset{c}{0\textwidth}{c}{0\textwidth}{w}{
                \includegraphics[width=0.7\textwidth]{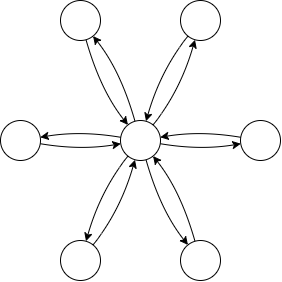}
            }
            \caption{$O_m$ $m$ even}
        \end{subfigure}
        \begin{subfigure}[b]{0.3\textwidth}
            \centering
            \stackinset{c}{0\textwidth}{c}{0\textwidth}{w}{
                \includegraphics[width=0.7\textwidth]{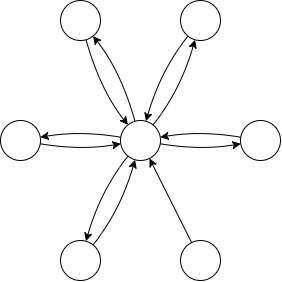}
            }
            \caption{$O_m$ $m$ odd}
        \end{subfigure}
        \begin{subfigure}[b]{0.3\textwidth}
            \centering
            \stackinset{c}{0\textwidth}{c}{0\textwidth}{w}{
                \includegraphics[width=0.7\textwidth]{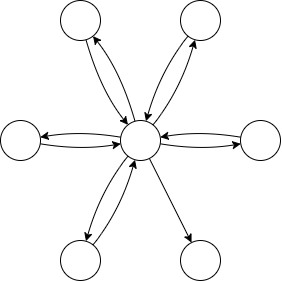}
            }
            \caption{$O_m^T$ $m$ odd}
        \end{subfigure}
        \captionsetup{justification=centering}
        \caption{Optimal digraphs for $m = 12$ and $m = 11$. The one on the left verifies $\Phi (G) = 6*6 + (1 * 1) \cdot 6 = $42. Those on the right verify $\Phi (G) = 5 * 6 + (1 * 1) \cdot 5 + 1 * 0= 35$}
        \label{fig:bigroot:1}
    \end{figure}
    These are the root digraphs of the line digraphs shown in Figure~\ref{fig:maxlinedigraphe}.
    For theses digraphs, by Theorem~\ref{the:maxlinedigraphe}, $\Phi (O_m) = \frac{(m - 1)^2}{4}+ m - 1$ if $m$ is odd and $\Phi (O_m) = \Phi (O_m^T) = \frac{m^2}{4}+ \frac{m}{2}$ if it's even. Let us show that every optimal digraph with $m$ arcs is isomorphic to $O_m$ or $O_m^T$. Note that $O_m = O_m^T$ if $m$ is even.

    \begin{lemma}
    \label{lem:bigroot:11}
        If $G$ is optimal, then, for any arc $(u, v) \in A$,
        $$\degpreds{u} + \degsuccs{v} \geq \begin{cases}
            \frac{m}{2} + 1 & \text{if $m$ is even}\\
            \frac{m - 1}{2} & \text{if $m$ is odd}
        \end{cases}$$
    \end{lemma}
    \begin{proof}
        Consider the digraph $F = G - (u, v)$. Then $\Phi (F) = \Phi (G) - (\degpreds{u} + \degsuccs{v})$. There are $m - 1$ arcs in the digraph $F$. Suppose $m$ is even.
        By Theorem~\ref{the:maxlinedigraphe}:
        \begin{align*}
            \Phi (G) &= \frac{m^2}{4} + \frac{m}{2}\quad\text{and}\quad \Phi (F) \leq \frac{(m - 2)^2}{4} + m - 2\\
            \text{So}\qquad\degpreds{u} + \degsuccs{v} &\geq \frac{m^2}{4} + \frac{m}{2} - \left(\frac{(m - 2)^2}{4} + m - 2\right) \geq \frac{m}{2} + 1
        \end{align*}
        The same applies if $m$ is odd, by Theorem~\ref{the:maxlinedigraphe}:
        \begin{align*}
            \Phi (G) &= \frac{(m - 1)^2}{4} + m - 1 \quad\text{and}\quad \Phi (F) \leq \frac{(m - 1)^2}{4} + \frac{m - 1}{2}\\
            \text{So}\qquad \degpreds{u} + \degsuccs{v} &\geq \frac{(m - 1)^2}{4} + m - 1 - \left(\frac{(m - 1)^2}{4} + \frac{m - 1}{2}\right) \geq \frac{m - 1}{2}\qedhere
        \end{align*}
    \end{proof}
    \begin{corollary}
    \label{cor:bigroot:19}
        If $G$ is optimal, for any pair of arcs $(u, v)$ and $(x, y)$ then, if $F = G - (u, v)$, $\degpreds[G]{u} + \degsuccs[G]{v} + \degpreds[F]{x}+ \degsuccs[F]{y} \geq m$
    \end{corollary}
    \begin{proof}
        This is an application of Lemma~\ref{lem:bigroot:11} first with arc $(u, v)$ in digraph $G$ and then with arc $(x, y)$ in digraph $F$.
    \end{proof}

    Lemmas~\ref{lem:bigroot:18} to \ref{lem:bigroot:20} show that if $G$ is optimal, it must contain a 2-length circuit.

        \begin{lemma}
        \label{lem:bigroot:18}
            If $G$ is optimal, with no 2-length circuit, then, for any pair of consecutive arcs $(u, v)$ and $(v, x)$ we have $(x, u) \in A$.
        \end{lemma}
        \begin{proof}
            Let $F = G - (u, v)$. Then $\degpreds[F]{v} = \degpreds[G]{v} - 1$. Since there is no 2-length circuit in $G$, then $x \neq u$ and $\degsuccs[F]{x} = \degsuccs[G]{x}$. 
            By Corollary~\ref{cor:bigroot:19}, $\beta = \degpreds[G]{u} + \degsuccs[G]{v} + \degpreds[G]{v} + \degsuccs[G]{x} - 1 \geq m$.
            For all $(s,t)\in A$, if $\alpha^{s,t} = \mathbbm{1}_{u=t} + \mathbbm{1}_{v=s} + \mathbbm{1}_{v=t} + \mathbbm{1}_{x=s} \leq 1$ then $\beta = \sum\nolimits_{(s,t)\in A} \alpha^{s,t} - 1 < m$, contradiction. Note that $\alpha^{s,t} \leq 2$ because $u \neq v$ and $v \neq x$. So $\exists (s,t)\in A$, $\alpha^{s,t} = 2$. There are three cases for arc $(s, t)$ : the arc $(v,u)$ or the arc $(x,v)$, which are excluded since there is no 2-length circuit in $G$, the last option being the arc $(x,u)$, which is the desired result.
        \end{proof}

        \begin{lemma}
        \label{lem:bigroot:17}
            If $G$ is optimal with no 2-length circuit, then every pair of arcs are incident on a same node.
        \end{lemma}
        \begin{proof}
            Suppose there are two arcs in $G$ with no common extremity, $(u, v)$ and $(x, y)$. Note that the degree of $x$ and $y$ in $G - (u, v)$ is unchanged. 
            By Corollary~\ref{cor:bigroot:19}, $\beta = \degpreds{u} + \degsuccs{v} + \degpreds{x} + \degsuccs{y} \geq m$.
            Let $\alpha^{s,t} = \mathbbm{1}_{u=t} + \mathbbm{1}_{v=s} + \mathbbm{1}_{x=t} + \mathbbm{1}_{y=s} \leq 2$. Since $\alpha^{u, v} = \alpha^{x, y} = 0$ and $\beta = \sum\nolimits_{(s,t)\in A} \alpha^{s,t}$ then $\exists (p,q), (s,t)\in A$, $\alpha^{p,q} = \alpha^{s,t} = 2$.
            There are four cases, the arcs $(v, u)$, $(y, x)$, $(v, x)$ and $(y, u)$. The first two cases are excluded because $G$ contains no 2-length circuit. So there is the circuit $(u, v, x, y, u)$ in $G$. 
            According to Lemma~\ref{lem:bigroot:18}, the arcs $(x, u)$ and $(u, x)$ are in $G$, which is excluded by hypothesis.
        \end{proof}

        \begin{lemma}
        \label{lem:bigroot:20}
            If $G$ is optimal, then there is at least one 2-length circuit in $G$.
        \end{lemma}
        \begin{proof}
            It is assumed that $m \geq 4$. The cases $m = 2$ and $m = 3$ can be handled by exhaustive enumeration. Assume that there is no symmetric arcs in $G$. By Lemma~\ref{lem:bigroot:17}, two arcs of $G$ have a common node. If there is a triangle $(u, v, w)$ in the undirected graph underlying $G$, any other arc must be incident on at least two nodes of the triangle to have a node in common with each arc. Since $m \geq 4$ then there is a 2-length circuit in $G$.

            If there is no triangle in the undirected graph underlying $G$, and all pairs of arcs have a node in common, then all arcs are incident to the same node $w$. If there is no 2-length circuit in $G$, by Lemma~\ref{lem:bigroot:18}, $w$ is either a \emph{source} ($\degpreds{w} = 0$) or a \emph{sink} ($\degsuccs{w} = 0$), but then $\Phi (G) = 0$ and $G$ is not optimal.
        \end{proof}

    Lemmas~\ref{lem:bigroot:16} to \ref{lem:bigroot:23}  show that if $G$ is optimal with $m \geq 7$ arcs and there is a 2-length circuit, then $G$ is isomorphic to $O_m$ or $O_m^T$ (Figure~\ref{fig:bigroot:1}).

    \begin{lemma}
    \label{lem:bigroot:16}
        If $G$ is optimal and there is a 2-length circuit $(u, v, u)$ then every arc is incident to $u$ or $v$.
    \end{lemma}
    \begin{proof}
        Let $F = G - (u, v)$. Then $\degpreds[F]{v} = \degpreds[G]{v} - 1$ and $\degsuccs[F]{u} = \degsuccs[G]{v} - 1$.
        By Corollary~\ref{cor:bigroot:19}, 
        $\degpreds[G]{u} + \degsuccs[G]{v} + \degpreds[G]{v} + \degsuccs[G]{u} - 2 \geq m$.
        Let $\alpha^{s,t} = \mathbbm{1}_{u=t} + \mathbbm{1}_{v=s} + \mathbbm{1}_{v=t} + \mathbbm{1}_{u=s}$. We have $\alpha^{u,v} = \alpha^{v,u} = 2$ and $\forall (s,t)\in A \setminus \{(u,v),(v,u)\}$, $\alpha^{s,t} \leq 1$
        (otherwise it would be a loop or $G$ would be a multigraph).
        We have $\sum_{(s,t)\in A} \alpha^{s,t} \geq m + 2$, so $\forall (s,t)\in A \setminus \{(u,v),(v,u)\}$, $\alpha^{s,t} = 1$
        and every arc is incident on $u$ or $v$.
    \end{proof}

    \begin{lemma}
    \label{lem:bigroot:22.0}
        If $G$ is optimal with $m \geq 7$
        
        $\bullet$ and there is a 2-length circuit $C_1 = (u, v, u)$
        
        $\bullet$ and there is another 2-length circuit $C_2$
        
        $\bullet$ and all nodes of $G$ are neighbors of $u$ or all nodes of $G$ are neighbors of $v$\\
        Then all arcs are incident to $u$ or all arcs are incident to $v$.
    \end{lemma}
    \begin{proof}
        $C_2$ arcs are incident to $u$ or $v$ by Lemma~\ref{lem:bigroot:16}. Without loss of generality, let assume that $C_2 = (u,x,u)$.
        Then by Lemma~\ref{lem:bigroot:16}, every arc is incident to $u$ or $v$ and every arc is incident to $u$ or $x$ so an arc that is not incident to $u$ is necessarily $(v, x)$ or $(x, v)$. There can therefore be no more than 2 arcs that are not incident to $u$ by Lemma~\ref{lem:bigroot:16}. If $m \geq 7$, there is at least one other arc incident to $u$. So if we consider the digraph $F$ where $(v, x)$ and $(x, v)$ are replaced by $(u, y)$ and $(y, u)$ where $y$ is a new node, we verify that $\Phi (F) > \Phi (G)$, which is excluded by optimality of $G$.
    \end{proof}
    \begin{lemma}
    \label{lem:bigroot:22.1}
         If $G$ is optimal, with a unique 2-length circuit $(u, v, u)$ and all nodes are neighbors of $u$ or all nodes are neighbors of $v$, then all arcs are incident to $u$ or all arcs are incident to $v$.
    \end{lemma}

    \begin{proof}
        Let assume that all nodes are neighbors of $u$ (case of $v$ is symmetrical).

        Let separate the nodes of $G - (u, v) $ into six sets:

        \begin{minipage}{0.49\linewidth}
            \begin{itemize}[leftmargin=0pt]
                \item $X_{uv}$ successor nodes of $u$\\and predecessors of $v$.
                \item $X_{vu}$ successor nodes of $v$\\and predecessors of $u$.
                \item $\Gamma^-_{uv}$ predecessor nodes of $u$ and $v$.
            \end{itemize}
        \end{minipage}
        \begin{minipage}{0.49\linewidth}
            \begin{itemize}[leftmargin=0pt]
                \item $\Gamma^-_{u}$ predecessor nodes of $u$\\but not adjacent to $v$.
                \item $\Gamma^+_{u}$ successor nodes of $u$\\but not adjacent to $v$.
                \item $\Gamma^+_{uv}$ successor nodes of $u$ and $v$.
            \end{itemize}
        \end{minipage}\\ \\
        
        Since there is no 2-length circuit in $G$ except $(u,v,u)$, there are no other categories. Note that all the nodes in these 6 categories are of degree 1 (if they are adjacent only to $u$) or 2 (if they are adjacent to $u$ and $v$). 

        Let now consider the digraph $F$ where
        \begin{itemize}
            \item for $x \in X_{uv}$, we replace the arc $(x, v)$ by $(x, u)$
            \item for $x \in X_{vu}$, we replace the arc $(v, x)$ by $(u, x)$
            \item for $x \in \Gamma^-_{uv}$, we replace the arc $(x, v)$ by arc $(y, u)$ where $y$ is a new node.
            \item for $x \in \Gamma^+_{uv}$, we replace the arc $(v, x)$ by arc $(u, z)$ where $z$ is a new node.
        \end{itemize}
        Let $\Gamma = \Gamma^-_{uv} \cup \Gamma^+_{uv} \cup \Gamma^-_{u} \cup \Gamma^+_{u}$. So we have
        \begin{align*}
            \Phi (G) &= \degpreds{u} \cdot \degsuccs{u} + \degpreds{v} \cdot \degsuccs{v} + \sum\limits_{x \in X_{uv}} \degpreds{x} \cdot \degsuccs{x} + \sum\limits_{x \in X_{vu}} \degpreds{x} \cdot \degsuccs{x} + \sum_{x \in \Gamma} \degpreds{x} \cdot \degsuccs{x}\\
            \Phi (G) &= \degpreds{u} \cdot \degsuccs{u} + \degpreds{v} \cdot \degsuccs{v} + \sum\limits_{x \in X_{uv}} 1 + \sum\limits_{x \in X_{vu}} 1 + 0\\
            \intertext{In $F$, $v$ is now adjacent only to $u$ with the 2-length circuit. All other arcs incident to $v$ are now incident to $u$.}
            \Phi (F) &= (\degpreds{u} + \degpreds{v} - 1) \cdot (\degsuccs{u} + \degsuccs{v} - 1) + 1 \cdot 1 + \sum\limits_{x \in X_{uv}} 1 + \sum\limits_{x \in X_{vu}} 1\\
            \Phi (F) &= \Phi (G) + \degpreds{u} \cdot \degsuccs{v} + \degpreds{v} \cdot \degsuccs{u} - \degpreds{v} - \degsuccs{v} - \degpreds{u} - \degsuccs{u} + 2\\
            \Phi (F) &= \Phi (G) + (\degpreds{u} \cdot \degsuccs{v} + 1 - \degpreds{u} - \degsuccs{v}) + (\degpreds{v} \cdot \degsuccs{u} + 1 - \degpreds{v} - \degsuccs{u})
        \end{align*}

        So $F$ is optimal and since $\Phi (G) = \Phi (F)$ then
        
        \centering{$(\degpreds{u} \cdot \degsuccs{v} + 1 - \degpreds{u} - \degsuccs{v}) + (\degpreds{v} \cdot \degsuccs{u} + 1 - \degpreds{v} - \degsuccs{u}) = 0$}

        Note that for any pair of positive non-zero integers, we have $xy + 1\geq x + y$. Equality exists only if $x = 1$ or $y = 1$. We deduce from this inequality that $\degpreds{u} = 1$ or $\degsuccs{v} = 1$ and that $\degpreds{v} = 1$ or $\degsuccs{u} = 1$.
        If $\degpreds{v} = \degsuccs{v} = 1$ then the lemma is proved.
        If $\degpreds{u} = \degsuccs{u} = 1$ then $u$ has only $v$ as a neighbor, the digraph contains only these two nodes and the lemma is proved.
        If $\degpreds{u} = \degpreds{v} = 1$ then among the neighbors of $u$ subsist the sets $\Gamma^+_{uv}$ and $\Gamma^+_{u}$. If $\Gamma^+_{uv} = \emptyset$ then the lemma is proved. So let consider a node $x \in \Gamma^+_{uv}$ and $G"$ the digraph where the arc $(v, x)$ is replaced by $(x, u)$. Let $\Sigma = \sum_{y \in V\backslash\{x, u, v\}} \degsuccs{y} \degpreds{y}$. 
        \begin{align*}
            \Phi (G) &= \degpreds{u} \cdot \degsuccs{u} + \degpreds{v} \cdot \degsuccs{v} + \degpreds{x} \cdot \degsuccs{x} + \Sigma = 1 \cdot \degsuccs{u} + 1 \cdot \degsuccs{v} + 2 \cdot 0 + \Sigma\\
            \Phi (G") &= (\degpreds{u} + 1) \cdot \degsuccs{u} + \degpreds{v} \cdot (\degsuccs{v} - 1) + (\degpreds{x} - 1) \cdot (\degsuccs{x} + 1) + \Sigma\\
            \Phi (G") &= \Phi (G) + \degsuccs{u} - \degpreds{v} + (\degpreds{x} - 1) \cdot (\degsuccs{x} + 1) = \Phi (G) + \degsuccs{u} - 1 + 1
        \end{align*}

        This contradicts the optimality of $G$. The case $\degsuccs{u} = 1$ and $\degsuccs{v} = 1$ is symmetrical. The lemma has been proved.
    \end{proof}

    \begin{lemma}
    \label{lem:bigroot:23}
        If $G$ is optimal with $m \geq 7$ and a 2-length circuit $(u, v, u)$ then every arc is incident to $u$ or every arc is incident to $v$.
    \end{lemma}
    \begin{proof}
        By Lemma~\ref{lem:bigroot:16}, every arc is incident to $u$ or $v$.
        If any node is neighbor of $u$ or any node is neighbor of $v$ then according to Lemma~\ref{lem:bigroot:22.0}~and~\ref{lem:bigroot:22.1}, the lemma is proved.
        Else consider $X = \{x_1, x_2, \dots, x_p\}$ the set of neighbors of $u$ which are not neighbors of $v$ and $Y = \{y_1, y_2, \dots, y_q\}$ the set of $v$ which are not neighbors of $u$. Suppose $q \leq p$. Let construct the digraph $F$ where the arcs $(y_i, v)$ (respectively $(v, y_i)$) are replaced by $(x_i, v)$ (respectively $(v, x_i)$). By optimality of $G$, for all $i \leq q$, $x_i$ and $y_i$ are sources or $x_i$ and $y_i$ are sinks. Otherwise, $\degpreds[F]{x_i}\cdot\degsuccs[F]{x_i} = 1$ and so $\Phi (F) > \Phi (G)$. Without loss of generality, let assume that $x_1$ and $y_1$ are sources. So there are the arcs $(x_1, u)$ and $(x_1, v)$ in $F$. Since $G$ is optimal, so is $F$. Note that in $F$ every node is neighbor of $u$ or every node is neighbor of $v$ so by Lemma~\ref{lem:bigroot:22.0}~and~\ref{lem:bigroot:22.1}, every arc of $F$ is incident to $u$ or every arc of $F$ is incident to $v$. This contradicts the existence of $(x_1, u)$ and $(x_1, v)$.
    \end{proof}

    \begin{theorem}
    \label{the:bigroot:24}
        If $G$ is optimal and $m \geq 7$ then $G$ is isomorphic to $O_m$ or $O_m^T$.
    \end{theorem}
    \begin{proof}
        By Lemma~\ref{lem:bigroot:20}, there are two symmetric arcs $(u, v)$ and $(v, u)$ in $G$. By Lemma~\ref{lem:bigroot:23}, every arc is incident to $u$ or every arc is incident to $v$. Let assume without loss of generality that they are incident to $u$.
        So $G$ is a star where every node is adjacent to $u$ and is connected to $u$ by one arc or two symmetrical arcs.
        Note that if there is a source and a sink  in $G$ connected to $u$ then we can merge these nodes to increase $\Phi (G)$. So there are only sources or only sinks in $G$.
        If there are no source and no sink in $G$, then we are in the case of the first digraph in Figure~\ref{fig:bigroot:1}.
        If there are only sources in $G$, we can deduce that $\degsuccs{u} \leq \degpreds{u}$. If there is only one source in $G$, then we are in the case of the last digraph in Figure~\ref{fig:bigroot:1}. And if there are two sources $x$ and $y$ in $G$, then by replacing $(y, u)$ by $(u, x)$, we obtain a digraph $F$ such that 
        \noindent\textbf{}$$\Phi (F) = \Phi (G) - \degsuccs{u} + \degpreds{u} + 1$$
        Since $\degsuccs{u} \leq \degpreds{u}$, we have a contradiction with the optimality of $G$.
        Similarly, if there are only sinks in $G$, we are in the case of the second digraph in Figure~\ref{fig:bigroot:1}.
    \end{proof}

\section{Conclusion}

    In this paper, we have shown that the maximum number of arcs in a line digraph with $n$ nodes is $\left(\frac{n}{2}\right)^2 + \frac{n}{2}$ arcs if $n$ is even, and $\left(\frac{n - 1}{2}\right)^2 + n - 1$ arcs otherwise.
    We have also shown that , for $n \geq 7$, the only line digraphs with so many arcs are those shown in Figure~\ref{fig:maxlinedigraphe}.

\DeclareEmphSequence{\itshape}

\bibliographystyle{acm}
\bibliography{references}

\section{Annexe}
    \DeclareEmphSequence{\bfseries\itshape}
    \newcommand\widthAnnexe{5.5cm}
    Here we present the initialization of the proof of Theorem~\ref{the:maxlinedigraphe}, for line digraphs of order $\leq 6$. We show this using the root digraph and we want to maximize $\Phi (G) = \sum_{v \in V} \degsuccs{v}\cdot \degpreds{v}$.
    All digraphs are assumed to be connected. We call $u$ a \emph{best node}, if $\forall v\in V, \degsuccs{v}\cdot \degpreds{v} \leq \degsuccs{u}\cdot \degpreds{u}$.

    \begin{lemma}
    \label{lem:centralnodefull}
        If all arcs of a digraph $G$ are incident to the best node $b$, with $x$ incoming arcs and $y$ outgoing arcs, then $\Phi(G) \leq x*y + min(x,y)$, with equality if there are $max(x,y) + 1$ nodes in $G$.
    \end{lemma}
    \begin{proof}
        Let assume that $x \geq y$, the other case is symmetrical. There are at least $x$ neighbors at b. If there are $x + 1$ nodes in $G$ then there is an outgoing arc from all the neighbors of $b$, and there is an incoming arc to $y$ of them. Thus $\Phi (G) = x*y + min(x,y)$. If there are $x + 1 + k$ nodes in $G$ $(0 \leq k \leq y)$, then there is an incoming arc to $k$ nodes with $\degsuccs{} = 0$, an outgoing arc from $x - y + k$ neighbors of $b$ with $\degpreds{} = 0$, and there is a 2-length circuit with $y - k$ of them. Thus $\Phi (G) = x*y + y - k \leq x*y + y$.
    \end{proof}
    
    We will treat each digraph according to its number of arcs and the value of the best node.
 
    \noindent\textbf{If there are two arcs in $G$.} $\forall v \in V, \degsuccs{v}\cdot \degpreds{v} \leq 1$. If there are two nodes $\{u,v\}$ in $G$, then there must be the arcs $(u,v)$ and $(v,u)$, thus $\Phi (G) = 2$ (figure~\ref{fig:root2}). If there are three nodes in $G$, then $\Phi (G) \leq 1$.

    \begin{figure}[ht!]
         \centering
         \includegraphics[width=\widthAnnexe]{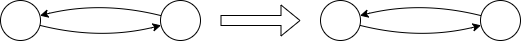}
         \captionsetup{justification=centering}
         \caption{root digraph with four arcs and their line digraph}
         \label{fig:root2}
    \end{figure}

    \noindent\textbf{if there are three arcs in $G$.} $\forall v \in V, \degsuccs{v}\cdot \degpreds{v} \leq 2$.

    If a best node $u$ is such that $\degsuccs{u}\cdot \degpreds{u} = 2$, then $\degsuccs{u} = 2$ and $\degpreds{u}=1$ (resp $\degsuccs{u} = 1$ and $\degpreds{u}=2$). In this case by Lemma~\ref{lem:centralnodefull}, $\Phi (G) = 3$ if there are three nodes in $G$ (figure~\ref{fig:root3a}, resp figure~\ref{fig:root3b}) and $\Phi (G) < 3$ if there are more nodes.

    If a best node $u$ is such that $\degsuccs{u}\cdot \degpreds{u} = 1$, and $\Phi (G) \geq 3$, then there are at least three nodes with $\degsuccs{}\cdot \degpreds{} = 1$ in $G$, this only happens with the 3-circuit (figure~\ref{fig:root3c}). In this case, it requires more than three arcs to have $\Phi (G) > 3$.
    \\

    \begin{minipage}{0.49\linewidth}
        \centering
        \includegraphics[width=\widthAnnexe]{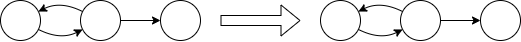}
        \captionsetup{justification=centering}
        \captionof{figure}{}
        \label{fig:root3a}
        \includegraphics[width=\widthAnnexe]{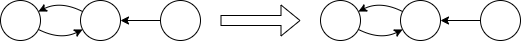}
        \captionsetup{justification=centering}
        \captionof{figure}{}
        \label{fig:root3b}
    \end{minipage}
    \begin{minipage}{0.49\linewidth}
        \centering
        \includegraphics[width=\widthAnnexe]{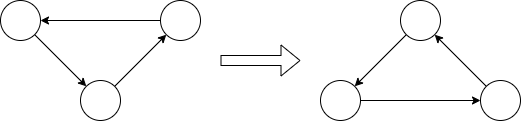}
        \captionsetup{justification=centering}
        \captionof{figure}{}
        \label{fig:root3c}
    \end{minipage}

    \begin{center}
         root digraph with four arcs and their line digraph
    \end{center}

    \noindent\textbf{if there are four arcs in $G$}. $\forall v \in V, \degsuccs{v}\cdot \degpreds{v} \leq 4$.

    If a best node $u$ is such that $\degsuccs{u}\cdot \degpreds{u} = 4$, then $\degsuccs{u} = 2$ and $\degpreds{u}=2$. In this case by Lemma~\ref{lem:centralnodefull}, $\Phi (G) = 6$ if there are three nodes in $G$ (figure~\ref{fig:root4}) and $\Phi (G) < 6$ if there are more nodes.

    If a best node $u$ is such that $\degsuccs{u}\cdot \degpreds{u} = 3$, then $\degsuccs{u} = 3$ and $\degpreds{u}=1$ (resp $\degsuccs{u} = 1$ and $\degpreds{u}=3$). In this case by Lemma~\ref{lem:centralnodefull}, $\Phi (G) \leq 4 < 6$

    If a best node $u$ is such that $\degsuccs{u}\cdot \degpreds{u} = 2$, then $\degsuccs{u} = 2$ and $\degpreds{u}=1$ (resp $\degsuccs{u} = 1$ and $\degpreds{u}=2$). Let assume that $\Phi (G) \geq 6$. If there is only one best node in $G$, then there are at least four nodes with $\degsuccs{}\cdot \degpreds{} = 1$, then $\sum_{v \in V} \degsuccs{v} > 4$, impossible. If there are two best nodes in $G$, then there are at least two nodes with $\degsuccs{}\cdot \degpreds{} = 1$, then $\sum_{v \in V} (\degsuccs{v} + \degpreds{v}) > 2*4$, impossible. If there are at least three best nodes in $G$, then $\sum_{v \in V} (\degsuccs{v} + \degpreds{v}) > 2*4$, impossible.

    If a best node $u$ is such that $\degsuccs{u}\cdot \degpreds{u} = 1$ and $\Phi (G) \geq 6$, then there are at least six nodes with $\degsuccs{}\cdot \degpreds{} = 1$ in $G$, then $\sum_{v \in V} \degsuccs{v} > 4$, impossible.

    \begin{figure}[ht!]
         \centering
         \includegraphics[width=\widthAnnexe]{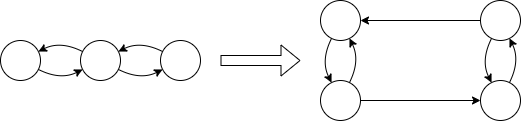}
         \captionsetup{justification=centering}
         \caption{root digraph with four arcs and their line digraph}
         \label{fig:root4}
    \end{figure}

    \noindent\textbf{if there are five arcs in $G$.} $\forall v \in V, \degsuccs{v}\cdot \degpreds{v} \leq 6$.

    If a best node $u$ is such that $\degsuccs{u}\cdot \degpreds{u} = 6$, then $\degsuccs{u} = 3$ and $\degpreds{u}=2$ (resp $\degsuccs{u} = 2$ and $\degpreds{u}=3$). In this case by Lemma~\ref{lem:centralnodefull}, $\Phi (G) = 8$ if there are four nodes in $G$ (figure~\ref{fig:root5a}, resp figure~\ref{fig:root5b}) and $\Phi (G) < 8$ if there are more nodes.

    If a best node $u$ is such that $\degsuccs{u}\cdot \degpreds{u} = 4$, then $\degsuccs{u} = 4$ and $\degpreds{u}=1$ (resp $\degsuccs{u} = 1$ and $\degpreds{u}=4$) or $\degsuccs{u} = 2$ and $\degpreds{u}=2$. If $\degsuccs{u} = 4$ (resp $\degpreds{u}=4$), then by Lemma~\ref{lem:centralnodefull}, $\Phi (G) \leq 5 < 8$. If $\degsuccs{u} = 2$ and $\degpreds{u}=2$, then there is only one best node in $G$, otherwise it would have at least six arcs since two nodes can only share two arcs. For the same reason, there is not a node $v$ such that $\degsuccs{v}\cdot \degpreds{v} = 3$ in $G$. If there is a node $v$ such that $\degsuccs{v}\cdot \degpreds{v} = 2$ in $G$, then it must share two arcs with $u$. Then there are two arcs between $u$ and $v$, an incoming arc on $u$, an outgoing arc from $u$, and an incident arc on $v$. If there are three nodes in $G$, then $\Phi (G) = 4 + 2 + 2 = 8$ (figure~\ref{fig:root5c}). If there are four nodes in $G$, then $\Phi (G) \leq 4 + 2 + 1 + 0 < 8$. If there are five nodes in $G$, then $\Phi (G) = 6 < 8$. If there is no node such that $\degsuccs{}\cdot \degpreds{} = 2$ in $G$ and $\Phi (G) \geq 8$, then there are at least four nodes with $\degsuccs{}\cdot \degpreds{} = 1$, and $\sum_{v \in V} \degsuccs{v} > 5$, impossible.

    If a best node $u$ is such that $\degsuccs{u}\cdot \degpreds{u} = 3$, then $\degsuccs{u} = 3$ and $\degpreds{u}=1$ (resp $\degsuccs{u} = 1$ and $\degpreds{u}=3$) and there is only one best node in $G$, otherwise it would have at least six arcs since two nodes can only share two arcs. Let assume $\Phi (G) \geq 8$. If there is one node such that $\degsuccs{}\cdot \degpreds{} = 2$ in $G$, then there are at least three nodes with $\degsuccs{}\cdot \degpreds{} = 1$, and $\sum_{v \in V} (\degsuccs{v} + \degpreds{v}) > 2*5$. If there are two nodes such that $\degsuccs{}\cdot \degpreds{} = 2$ in $G$, then there is at least one node with $\degsuccs{}\cdot \degpreds{} = 1$, and $\sum_{v \in V} (\degsuccs{v} + \degpreds{v}) > 2*5$. If there are at least three nodes such that $\degsuccs{}\cdot \degpreds{} = 2$ in $G$, then $\sum_{v \in V} (\degsuccs{v} + \degpreds{v}) > 2*5$. If there is no node such that $\degsuccs{}\cdot \degpreds{} = 2$ in $G$ and $\Phi (G) \geq 8$, then there is at least five nodes with $\degsuccs{}\cdot \degpreds{} = 1$, and $\sum_{v \in V} \degsuccs{v} > 5$, impossible.

    If a best node $u$ is such that $\degsuccs{u}\cdot \degpreds{u} = 2$, then $\degsuccs{u} = 2$ and $\degpreds{u}=1$ (resp $\degsuccs{u} = 1$ and $\degpreds{u}=2$). Let assume that $\Phi (G) \geq 8$. If there are at most two best nodes in $G$, then there are at least four nodes with $\degsuccs{}\cdot \degpreds{} = 1$, and $\sum_{v \in V} (\degsuccs{v} + \degpreds{v}) > 2*5$, impossible. If there are three best nodes in $G$, then there are at least two nodes with $\degsuccs{}\cdot \degpreds{} = 1$, and $\sum_{v \in V} (\degsuccs{v} + \degpreds{v}) > 2*5$, impossible. If there are at least four best nodes in $G$, then $\sum_{v \in V} (\degsuccs{v} + \degpreds{v}) > 2*5$, impossible.

    If a best node $u$ is such that $\degsuccs{u}\cdot \degpreds{u} = 1$ and $\Phi (G) \geq 8$, then there are at least eight nodes with $\degsuccs{}\cdot \degpreds{} = 1$ in $G$, and $\sum_{v \in V} \degsuccs{v} > 5$, impossible.

    \begin{figure}[ht!]
        \centering
        \hfill
        \begin{subfigure}[b]{0.49\textwidth}
            \centering
            \includegraphics[width=\widthAnnexe]{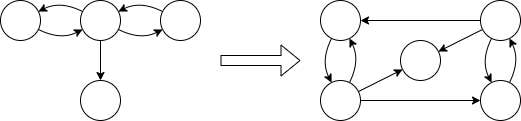}
            \caption{}
            \label{fig:root5a}
        \end{subfigure}
         \hfill
         \begin{subfigure}[b]{0.49\textwidth}
            \centering
            \includegraphics[width=\widthAnnexe]{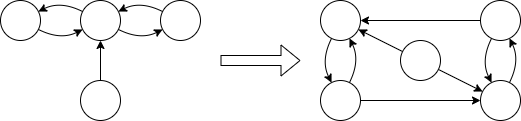}
            \caption{}
            \label{fig:root5b}
         \end{subfigure}
         \hfill
         \begin{subfigure}[b]{0.49\textwidth}
            \centering
            \includegraphics[width=\widthAnnexe]{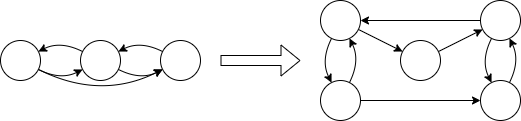}
            \caption{}
            \label{fig:root5c}
         \end{subfigure}
         \hfill
         \captionsetup{justification=centering}
         \caption{root digraph with five arcs and their line digraph}
         \label{fig:root5}
    \end{figure}

    \noindent\textbf{if there are six arcs in $G$.} $\forall v \in V, \degsuccs{v}\cdot \degpreds{v} \leq 9$. 

    If a best node $u$ is such that $\degsuccs{u}\cdot \degpreds{u} = 9$, then $\degsuccs{u} = 3$ and $\degpreds{u} = 3$. In this case by Lemma~\ref{lem:centralnodefull}, $\Phi (G) = 12$ if there are four nodes in $G$ (figure~\ref{fig:root6a}) and $\Phi (G) < 12$ if there are more nodes.

    If a best node $u$ is such that $\degsuccs{u}\cdot \degpreds{u} = 8$, then $\degsuccs{u} = 4$ and $\degpreds{u}=2$ (resp $\degsuccs{u} = 2$ and $\degpreds{u} = 4$). In this case by Lemma~\ref{lem:centralnodefull}, $\Phi (G) \leq 10 < 12$.

    If a best node $u$ is such that $\degsuccs{u}\cdot \degpreds{u} = 6$, then $\degsuccs{u} = 3$ and $\degpreds{u}=2$ (resp $\degsuccs{u} = 2$ and $\degpreds{u} = 3$) and there is no other node such that $\degsuccs{}\cdot \degpreds{} \geq 3$ in $G$, otherwise it would have at least seven arcs since two nodes can only share two arcs.
    Let assume $\Phi (G) \geq 12$. If there is one node such that $\degsuccs{}\cdot \degpreds{} = 2$ in $G$, then there are at least four nodes with $\degsuccs{}\cdot \degpreds{} = 1$, and $\sum_{v \in V} (\degsuccs{v} + \degpreds{v}) > 2*6$. If there are two nodes such that $\degsuccs{}\cdot \degpreds{} = 2$ in $G$, then there are at least two nodes with $\degsuccs{}\cdot \degpreds{} = 1$, and $\sum_{v \in V} (\degsuccs{v} + \degpreds{v}) > 2*6$. If there are at least three nodes such that $\degsuccs{}\cdot \degpreds{} = 2$, then $\sum_{v \in V} (\degsuccs{v} + \degpreds{v}) > 2*6$. If there is no node such that $\degsuccs{}\cdot \degpreds{} = 2$ in $G$ and $\Phi (G) \geq 12$, then there are at least six nodes with $\degsuccs{}\cdot \degpreds{} = 1$, and $\sum_{v \in V} \degsuccs{v} > 6$, impossible.

    If a best node $u$ is such that $\degsuccs{u}\cdot \degpreds{u} = 5$, then $\degsuccs{u} = 5$ and $\degpreds{u}=1$ (resp $\degsuccs{u} = 1$ and $\degpreds{u} = 5$). In this case by Lemma~\ref{lem:centralnodefull}, $\Phi (G) \leq 10 < 12$.

    If a best node $u$ is such that $\degsuccs{u}\cdot \degpreds{u} = 4$, then $\degsuccs{u} = 4$ and $\degpreds{u}=1$ (resp $\degsuccs{u} = 1$ and $\degpreds{u}=4$) or $\degsuccs{u} = 2$ and $\degpreds{u}=2$. If $\degsuccs{u} = 4$ (resp $\degpreds{u}=4$), then there is no other node such that $\degsuccs{}\cdot \degpreds{} \geq 2$ in $G$, otherwise it would have at least seven arcs since two nodes can only share two arcs. If $\Phi (G) \geq 12$, then there are at least eight nodes with $\degsuccs{}\cdot \degpreds{} = 1$ in $G$, and $\sum_{v \in V} \degsuccs{v} > 6$, impossible.
    If $\degsuccs{u} = 2$ and $\degpreds{u}=2$, then there are at most three nodes such that $\degsuccs{}\cdot \degpreds{} \geq 3$ in $G$, otherwise it would have at least seven arcs since two nodes can only share two arcs. If there are three best nodes in $G$, then $\Phi (G) = 12$ (figure~\ref{fig:root6a}). Else, if there is a node $v$ such that $\degsuccs{v}\cdot \degpreds{v} = 3$ in $G$, then it must share two arcs with $u$. So there are two arcs between $u$ and $v$, an incoming arc on $u$, an outgoing arc from $u$, and two arcs with the same orientation on $v$. If there are four nodes in $G$, then $\Phi (G) \leq 4 + 3 + 2 + 0 < 12$. If there are five nodes in $G$, then $\Phi (G) \leq 4 + 3 + 1 + 0 + 0 < 12$. If there are six nodes in $G$, then $\Phi (G) \leq 4 + 3 + 0 + 0 + 0 < 12$. If there are exactly two best nodes in $G$, then they must share two arcs and there must be at least four nodes in $G$. If there are four nodes in $G$, then $\Phi (G) = 4 + 4 + 1 + 1 < 12$. If there are five nodes in $G$, then $\Phi (G) = 4 + 4 + 1 + 0 + 0 < 12$. If there are six nodes in $G$, then $\Phi (G) = 4 + 4 + 0 + 0 + 0 + 0 < 12$. If there is only one best node in $G$, let assume $\Phi (G) \geq 12$. If there are one or two nodes such that $\degsuccs{}\cdot \degpreds{} = 2$ in $G$, then there are at least four nodes with $\degsuccs{}\cdot \degpreds{} = 1$, and $\sum_{v \in V} (\degsuccs{v} + \degpreds{v}) > 2*6$. If there are at least three nodes such that $\degsuccs{}\cdot \degpreds{} = 2$ in $G$, then $\sum_{v \in V} (\degsuccs{v} + \degpreds{v}) > 2*6$. If there is no node such that $\degsuccs{}\cdot \degpreds{} = 2$ in $G$ and $\Phi (G) \geq 12$, then there are at least eight nodes with $\degsuccs{}\cdot \degpreds{} = 1$, and $\sum_{v \in V} \degsuccs{v} > 6$, impossible.

    If a best node $u$ is such that $\degsuccs{u}\cdot \degpreds{u} = 3$, then $\degsuccs{u} = 3$ and $\degpreds{u}=1$ (resp $\degsuccs{u} = 1$ and $\degpreds{u}=3$) and there are at most three nodes such that $\degsuccs{}\cdot \degpreds{} = 3$ in $G$, otherwise it would have at least seven arcs since two nodes can only share two arcs. Let assume $\Phi (G) = 12$, if there are at least two best nodes $\{u,v\}$ in $G$, then they must share two arcs. So there are two arcs between $u$ and $v$. If $\degsuccs{u} = \degsuccs{v}$, then there are four, five or six nodes in $G$ and $\Phi (G) = 3 + 3 + 0 < 12$. If $\degsuccs{u} \neq \degsuccs{v}$, then there are four nodes in $G$, $\Phi (G) < 3 + 3 + 1 + 1 < 12$. If there are five nodes in $G$, then $\Phi (G) = 3 + 3 + 1 + 0 + 0 < 12$. If there are six nodes in $G$, then $\Phi (G) = 3 + 3 + 0 + 0 + 0 + 0 < 12$. Let assume that there is only one best node in $G$. If there are one or two nodes such that $\degsuccs{}\cdot \degpreds{} = 2$ in $G$, then there are at least five nodes with $\degsuccs{}\cdot \degpreds{} = 1$, and $\sum_{v \in V} (\degsuccs{v} + \degpreds{v}) > 2*6$. If there are at least three nodes such that $\degsuccs{}\cdot \degpreds{} = 2$ in $G$, then $\sum_{v \in V} (\degsuccs{v} + \degpreds{v}) > 2*6$. If there is no node such that $\degsuccs{}\cdot \degpreds{} = 2$ in $G$ and $\Phi (G) \geq 12$, then there are at least nine nodes with $\degsuccs{}\cdot \degpreds{} = 1$, and $\sum_{v \in V} \degsuccs{v} > 6$, impossible.

    If a best node $u$ is such that $\degsuccs{u}\cdot \degpreds{u} = 2$, then $\degsuccs{u} = 2$ and $\degpreds{u}=1$ (resp $\degsuccs{u} = 1$ and $\degpreds{u}=2$). Let assume that $\Phi (G) \geq 12$. If there are at most three best nodes in $G$, then there are at least six nodes with $\degsuccs{}\cdot \degpreds{} = 1$, and $\sum_{v \in V} \degsuccs{v} > 6$, impossible. If there are four best nodes in $G$, then there are at least four nodes with $\degsuccs{}\cdot \degpreds{} = 1$, and $\sum_{v \in V} \degsuccs{v} > 6$, impossible. If there are at least five best nodes in $G$, then $\sum_{v \in V} (\degsuccs{v} + \degpreds{v}) > 2*6$, impossible.

    If a best node $u$ is such that $\degsuccs{u}\cdot \degpreds{u} = 1$ and $\Phi (G) \geq 12$, then there are at least twelve nodes with $\degsuccs{}\cdot \degpreds{} = 1$ in $G$, then $\sum_{v \in V} \degsuccs{v} > 6$, impossible.

    \begin{figure}[ht!]
         \centering
         \hfill
         \begin{subfigure}[b]{0.49\textwidth}
            \centering
            \includegraphics[width=\widthAnnexe]{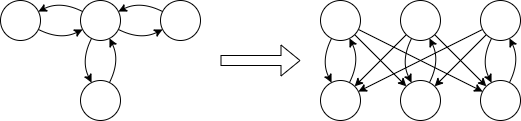}
            \caption{}
            \label{fig:root6a}
        \end{subfigure}
        \hfill
        \begin{subfigure}[b]{0.49\textwidth}
            \centering
            \includegraphics[width=\widthAnnexe]{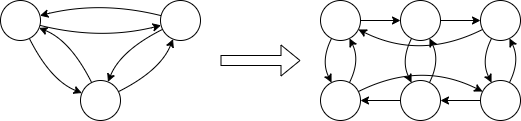}
            \caption{}
            \label{fig:root6b}
        \end{subfigure}
        \hfill
        \captionsetup{justification=centering}
        \caption{root digraph with six arcs and their line digraph}
        \label{fig:root6}
    \end{figure}

\end{document}